%% file: main.tex
\documentclass[11pt]{article}
\input{_style.sty}

\newcommand{\ma}{\ensuremath{\mathsf{MA}}}

\newcommand{\pspace}{\ensuremath{\mathsf{PSPACE}}}
\newcommand{\am}{\ensuremath{\mathsf{AM}}}
\newcommand{\pp}{\ensuremath{\mathsf{PP}}}

\newcommand{\stoqma}{\mathsf{StoqMA}}

\newcommand{\qma}{\mathsf{QMA}}
\newcommand{\expclass}{\mathsf{EXP}}
\newcommand{\nexp}{\ensuremath{\mathsf{NEXP}}}

\title{
    The Collapse of Unentangled Stoquastic \\
    Merlin-Arthur Proof Systems
}
\author{William Gay\thanks{{\tt University of Illinois, Urbana-Champaign}. {\tt whgay2@illinois.edu}. }  
   \and Fernando Granha Jeronimo\thanks{{\tt University of Illinois, Urbana-Champaign}. {\tt granha@illinois.edu}. }}

\date{September 30, 2026}

\begin{document}

\maketitle

\begin{abstract}
\noindent
Interference and entanglement are among the most fundamental properties of quantum mechanics. 
In this work, we investigate the power of interference in the role of detecting entanglement in the context of quantum Merlin-Arthur verification. 
In particular, we show that unentanglement gives no additional power to \textit{stoquastic} Merlin-Arthur proof systems, whose restriction to a classical gate set results in entrywise non-negative acceptance operators and non-negative optimal witnesses, therein removing destructive interference from the verification process.
Our results show that for any polynomial number of unentangled witnesses $k = k(n)$,
\[
\mathsf{StoqMA}(k) = \mathsf{StoqMA}. 
\]
The analytic core of this result is a new non-negative de Finetti theorem for \textit{separately bosonic} states. 
If a mixed state $\rho$ is supported on the tensor product of symmetric subspaces of $\mcH_1^{\otimes R} \otimes \mcH_2^{\otimes R} \otimes \cdots \otimes \mcH_k^{\otimes R}$ then there exists a separable state $\tau$ on $\mcH_1 \otimes \cdots \otimes \mcH_k$ such that
\[
\Tr\paren[\big]{M(\rho_{\mcH_1 \otimes \cdots \otimes \mcH_k} - \tau)} \le 
2 \sqrt{\frac{k\, \log_2\dim(\mcH_1 \otimes \cdots \otimes \mcH_k)}{R-1}},
\]
simultaneously for all non-negative effects $M$. 
Moreover, $\tau$ lies in the convex hull of non-negative pure product states. 
Thus, for any such operator $M$, one can define the separately symmetric lift $\mathrm{SS}_R(M)$ whose largest eigenvalue approximates the maximum acceptance probability of $M$ over product states within the same error bound. Such a lift is non-negative and can be approximated by a stoquastic verifier, giving that product-state constraint of $\mathsf{StoqMA}(k)$ verifiers can be absorbed into a polynomially larger one-witness $\mathsf{StoqMA}$ verification, while preserving an inverse polynomial completeness-soundness gap. Moreover, this transformation preserves completeness up to a constant factor, so combined with the result of~\cite{AharonovGriloLiu2025} yields that $\mathsf{StoqMA}(k)$ with near perfect completeness is contained in classical $\mathsf{MA}$. 
\end{abstract}
\newpage

\input{section1-introduction}

\input{section2-preliminaries}

\input{section3-positive-de-finetti}

\input{section4-stoquastic-collapse}

\section*{AI Disclosure}

We used ChatGPT 5.5 and 5.6 Pro primarily as technical proof accelerators for parts of this work. The core conceptual ideas, including the use of spectral methods and de Finetti techniques to improve upon the BKS Sum-of-Squares framework, are due to the authors. We also used ChatGPT 5.5 and 5.6 Pro as editorial assistants in preparing the manuscript, however, substantial human effort was also devoted to the exposition, organization, and presentation of the results. The authors take full responsibility for the results and have independently verified all claims, proofs, and references.

\section*{Acknowledgement}

We thank Yupan Liu and Pei Wu for pointing out the connection to perfect completeness and MA, as well, as a way to avoid an 
earlier dyadic truncation error. We thank Alex Grilo and Marios Rozos for also pointing out a similar connection to perfect
completeness and MA.

\bibliographystyle{alpha}
\bibliography{_refs}

\appendix
\input{sectionA-deferred-proofs}

\end{document}

%% file: section1-introduction.tex
\section{Introduction}

Two of the most fundamental resources in quantum mechanics, which are often attributed much of the separation between classical and quantum computation, are \textit{interference} and \textit{entanglement}. Interference refers to the idea that computational basis states are assigned complex amplitudes, rather than probabilities, and amplitudes with different phases may cancel each other out in a computation before a measurement is made. Entanglement refers to the idea that a quantum state on several registers need not be decomposable into separate states on the individual registers, even if the registers are physically separated. The purpose of this work is to study the interaction between interference and entanglement in quantum proof verification. Our results show that the absence of destructive interference, as found in stoquastic Merlin-Arthur proof systems, makes entanglement tame. That is, polynomially many unentangled proofs may be replaced by a single proof with only polynomial overhead. 

\textbf{Interference and Stoquasticity.}
The importance of interference becomes apparent in the problem of quantum Monte Carlo simulations. While efficient for many restricted classes of systems, Monte Carlo simulations of general quantum systems involve negative Boltzmann weights, which introduces exponential growth in the simulation time. Troyer and Wiese show that the generic quantum Monte Carlo sign problem is $\mathsf{NP}$-hard~\cite{TroyerWiese2005}. 
This motivates the study of \textit{stoquastic} Hamiltonians, which are Hamiltonians whose off-diagonal entries are real and non-positive when represented in some distinguished basis. 
This condition implies that the imaginary-time evolution of the Hamiltonian is entrywise non-negative. 
Moreover, the Perron-Frobenius theorem shows that such a Hamiltonian always has a ground state with non-negative amplitudes. 
Every Hamiltonian is diagonal, hence stoquastic, in some
basis, though that basis need not admit an efficient description. 
The computationally relevant \emph{curing} problem
instead asks whether a local Hamiltonian can be made stoquastic by an
efficiently describable product of local basis changes. 
Although certain families of Hamiltonians admit efficient algorithms for this problem, it is $\mathsf{NP}$-complete for many classes of Hamiltonians~\cite{KlassenTerhal2019,MarvianLidarHen2019}. 

The computational complexity of the stoquastic local Hamiltonian problem, with the stoquastic Hamiltonians already given in the relevant basis, was initially studied by Bravyi, DiVincenzo, Oliveira, \& Terhal and by Bravyi, Bessen, \& Terhal~\cite{BravyiDiVincenzoOliveiraTerhal2008,BravyiBessenTerhal2006}. 
The latter work introduced the complexity class $\mathsf{StoqMA}$ (sometimes called $\mathsf{StoqQMA}$) for which the stoquastic local Hamiltonian problem is complete. This class is a $\mathsf{QMA}$-style proof system in which the verifier is restricted to classical reversible gates, $\ket{0}$ and $\ket{+}$ ancilla, and a final measurement in the Hadamard basis. 
Such a verifier can be compressed into a single entrywise non-negative effect $M$. Consequently, an optimal witness may be chosen to have non-negative amplitudes. This formulation serves as the proof system analogue of removing destructive interference.

The class $\mathsf{StoqMA}$ occupies an interesting position between classical and quantum computation.  Indeed,~\cite{BravyiDiVincenzoOliveiraTerhal2008} showed that $\mathsf{MA} \subseteq \mathsf{StoqMA} \subseteq \mathsf{AM}$ and~\cite{BravyiBessenTerhal2006} improved the upper bound to $\mathsf{StoqMA} \subseteq \mathsf{SBP} \cap \mathsf{QMA}$, which places it simultaneously inside a classical counting and a quantum verification class that are not known to be comparable~\cite{BohlerGlasserMeister2006,MarriottWatrous2005}. In addition to the general stoquastic local Hamiltonian problem, the local Hamiltonian problem for the degree-3 transverse-field Ising model is also $\mathsf{StoqMA}$-complete~\cite{BravyiBessenTerhal2006,BravyiHastings2017}. Moreover $\mathsf{StoqMA}$ is one of the natural levels of classification for two local Hamiltonian problems in a quantum analogue of Schaefer's dichotomy theorem~\cite{CubittMontanaro2016}.

Despite the fact that $\mathsf{StoqMA}$ seems much more closely seated to classical complexity than $\mathsf{QMA}$, non-negativity alone does not immediately imply efficient simulation. In fact, there are oracle separations exhibiting subexponential quantum advantage for adiabatic evolution along a stoquastic path~\cite{Hastings2021,GilyenHastingsVazirani2021}.
Despite these challenges, there have been several results that relate stoquastic problems to fully classical computation. The frustration-free variant of stoquastic local Hamiltonian, stoquastic $k$-SAT, as well as guided stoquastic Hamiltonian are both $\mathsf{MA}$-complete~\cite{BravyiTerhal2010,Bravyi2015}. 
The stoquastic PCP program also relates gap amplification of frustration-free stoquastic local Hamiltonians to classical derandomization of $\mathsf{MA}$~\cite{AharonovGrilo2019}.
Aharonov, Grilo, and Liu show that sufficiently strong gap amplification for projection-uniform stoquastic local Hamiltonians would collapse $\mathsf{StoqMA}$ to $\mathsf{MA}$, and more specifically that projection-uniform stoquastic local Hamiltonian with near-perfect completeness and inverse-polynomial completeness-soundness gap is in $\mathsf{MA}$~\cite{AharonovGriloLiu2025}. 

\textbf{Entanglement and Multi-Proof Systems.}
The importance of entanglement enters proof-verification through quantum Merlin-Arthur systems with multiple unentangled witnesses.
The class $\mathsf{QMA}(k)$, introduced by Kobayashi, Matsumoto, and Yamakami, is a variant of $\mathsf{QMA}$ in which Arthur receives a quantum witness on $k$ registers with the promise that the witness is unentangled across the registers~\cite{KobayashiMatsumotoYamakami2009}. 
This classical analogue of this class collapses trivially back to a single witness, as multiple classical strings can always be concatenated into a single string, but concatenating quantum states removes that promise that the joint state is a tensor product. 
For the acceptance operator $M$ of such a verifier, rather than the unrestricted largest eigenvalue, the maximum acceptance probability becomes an optimization problem over product states
\[
\max_{\ket{u_1}, \cdots, \ket{u_k}} 
\matrixel{u_1 \otimes \cdots \otimes u_k}{M}{u_1 \otimes \cdots \otimes u_k}. 
\]
The power of this product state constraint has been studied from several directions. Having multiple unentangled witnesses gives surprisingly short proofs for 3-SAT and other NP-complete problems~\cite{BlierTapp2012,Beigi2010,AaronsonBeigiDruckerFeffermanShor2009,ChenDrucker2010,ChiesaForbes2013,LeGallNakagawaNishimura2012}. 
The first two of these works use two logarithmically sized proofs to obtain a protocol with inverse-polynomial completeness-soundness gap. Later works analyze the trade-off between number of witness and completeness-soundness gap. These results illustrate both sides of the unentangled promise: Honest provers can encode very structured proofs among many short states, but soundness depends on certifying consistency without allowing the proofs to become entangled. 
Harrow and Montanaro introduced the product test and used it to show that $\mathsf{QMA}(k) = \mathsf{QMA}(2)$ for any polynomially bounded $k \ge 2$~\cite{HarrowMontanaro2013}. Whether two witnesses are more powerful than one remains open, and the best known upper bound remains the trivial containment $\mathsf{QMA}(2) \subseteq \mathsf{NEXP}$. 
See~\cite{JeronimoLeighWu2026} for a survey of the complexity zoo of $\mathsf{QMA}(2)$ parameter regimes. 

Maximizing over product states is central to many optimization problems. Maximizing a linear function over product states, which is equivalent to maximizing over the convex set of separable states, is closely related to injective tensor norms, maximum output norms, and separability testing~\cite{HarrowMontanaro2013,GutoskiHaydenMilnerWilde2015}.
Separable variants of Hamiltonian problems make the latter connection concrete: Separable sparse Hamiltonian is $\mathsf{QMA}(2)$-complete, although, separable local Hamiltonian is in $\mathsf{QMA}$~\cite{ChaillouxSattath2012}. 
Even weak membership for the set of separable states is $\mathsf{NP}$-hard~\cite{Gurvits2003}. 
Thus, a central difficulty in obtaining better upper bounds for $\mathsf{QMA}(2)$-type problems is how to replace the product-state constraint by a tractable relaxation that still approximates separable states well enough for the verifier’s acceptance test.

Finite quantum de Finetti theorems show that marginals of permutation-invariant states are close to convex mixtures of product states, making symmetric-extension hierarchies a natural relaxation of the separability constraint~\cite{DohertyParriloSpedalieri2004,ChristandlKonigMitchisonRenner2007}. However, standard trace norm bounds have polynomial dependence on the Hilbert-space dimension, which is exponential in the witness length~\cite{ChristandlKonigMitchisonRenner2007}. 
By requiring approximation only against restricted classes of measurements, one can instead logarithmic dimension dependence, yielding collapses to $\mathsf{QMA}$ for several restricted multi-prover systems, including one-way-LOCC and Bell verifiers with constantly many proofs as well as fully one-way-LOCC protocols with polynomially many proofs~\cite{BrandaoChristandlYard2011,BrandaoHarrow2013,LiSmith2015}. These improvements are sensitive to the measurement class as, and for certain classes such as separable POVMs, the model retains the full power of $\mathsf{QMA}(2)$~\cite{HarrowMontanaro2013}. 

Liu and Wu initiated the systematic study of unentangled stoquastic proof systems~\cite{LiuWu2026}. 
In particular, they show that $\mathsf{NP}$ admits $\mathsf{StoqMA}(2)$ verifiers with short proofs, the inability for $\mathsf{StoqMA}(2)$ with exponentially small completeness-soundness gap to obtain perfect completeness unless $\mathsf{EXP} = \mathsf{NEXP}$, and establish fine-grained exponential upper bounds for $\mathsf{StoqMA}(k)$ for polynomially bounded $k$. The present work improves the latter upper bound all the way to $\mathsf{StoqMA}$, yielding a collapse of polynomially many proofs to a single proof.  Independently, Grilo and Rozos introduced stoquastic sparse-Hamiltonian problems and showed that the separable version of their stoquastic sparse-Hamiltonian problem is complete for the two-proof stoquastic class \cite{GR26}.  They explicitly ask for nontrivial upper bounds on $\mathsf{StoqMA}(2)$ beyond the trivial containment in $\mathsf{QMA}(2)$.

\subsection{Main Results}

We first state the approximation theorem underlying the collapse (see \Cref{sec:preliminaries} for precise definitions). 
A state is \textit{separately bosonic} if it is supported on the tensor product of the degree-$R$ symmetric subspaces of its $k$ blocks. 
We use $\rho_1$ to denote the tested marginal of a separately bosonic state and write $\mathrm{H}(\rho_1)$ for the entropy of the distribution given by measuring it in the computational basis. 

\begin{theorem}\label{thm:positive-de-Finetti}
Fix integers $k, R \in \N_{\ge 2}$, Hilbert spaces $\mcH_1, \ldots, \mcH_k$, and a separately-bosonic state
\[
\rho \in \mathrm{SepBos}(\mcH_1, \ldots, \mcH_k; R, \ldots, R).
\]
There exists a state $\tau \in \mathrm{Sep}_+(\mcH_1, \ldots, \mcH_k)$ such that
\[
\Tr\paren*{M\paren[\big]{\rho_1 - \tau}} \le
2\norm{M}_2\sqrt{\frac{k\, \mathrm{H}(\rho_1)}{R-1}},
\]
simultaneously for all non-negative self-adjoint operators $M$ on $\mcH_1 \otimes \cdots \otimes \mcH_k$.
\end{theorem}

The same state $\tau$ works for every non-negative test $M$. This uniformity permits the theorem to be applied directly to an unknown top eigenstate of a relaxation. 
Since $\mathrm{H}(\rho_1) \le \log_2\dim(\mcH_1 \otimes \cdots \otimes \mcH_k)$, the dimension enters only through its logarithm.
For an operator $M$ on $\mcH_1 \otimes \cdots \otimes \mcH_k$, let $\mathrm{SS}_R(M)$ be separately-symmetric extension of $M$ onto $\mcH_1^{\otimes R} \otimes \cdots \otimes \mcH_k^{\otimes R}$ given by applying $M$ on one copy of every block and projecting on both sides onto the tensor product of the $k$ symmetric subspaces. 
The next theorem shows that its largest eigenvalue is a monotone spectral relaxation of optimization over nonnegative
product states.

\begin{theorem}\label{thm:symmetric-lift}
Fix integers $k, R \in \N_{\ge 2}$ and Hilbert spaces $\mcH_1, \ldots, \mcH_k$.
Write $d := \dim(\mcH_1 \otimes \cdots \otimes \mcH_k)$.
For every non-negative self-adjoint operator $M$ on $\mcH_1 \otimes \cdots \otimes \mcH_k$,
\[
0 \le
\lambda_{\max}\paren*{\mathrm{SS}_R(M)}
- \max_{\substack{\ket{u_j} \in \bbS_+(\mcH_j) \\ j \in [k]}}
\matrixel{u_1 \otimes \cdots \otimes u_k}{M}{u_1 \otimes \cdots \otimes u_k}
\le
2\norm{M}_2\sqrt{\frac{k\max\set{1, \log_2(d)}}{R-1}}.
\]
In particular, if $\norm{M}_2 \le 1$, $\eps \in (0,1)$, and
\[
R \ge 1 + \frac{4k\max\set{1, \log_2(d)}}{\eps^2},
\]
then the upper bound is at most $\eps$.
\end{theorem}

We next turn to proof systems. Informally, $\mathsf{StoqMA}(k,q,c,s)$ is the class of promise problems having a stoquastic verifier with $k$ unentangled $q$-qubit witnesses, completeness at least $c$, and soundness at most $s$.
The parameters are functions of the input length. 
The notation $\mathsf{StoqMA}(k)$ takes the union over polynomial
witness lengths and inverse-polynomial gaps, and $\mathsf{StoqMA} = \mathsf{StoqMA}(1)$.

\begin{theorem}\label{thm:stoquastic-collapse}
Let $k, q, \Delta : \N \to \N_{\ge 1}$ be polynomially bounded and polynomial-time computable, with $k \ge 2$, and let $c,s : \N \to [\frac12,1]\cap\Q$ be polynomial-time computable functions satisfying
$c-s \ge 1/{\Delta}$.
Define
\[
R := 1 + \ceil{64k^2q\Delta^2}
\qquad\text{and}\qquad
q' := kRq.
\]
Then
\[
\mathsf{StoqMA}(k,q,c,s)
\subseteq
\mathsf{StoqMA}\paren*{1,q',\frac{1+c}{2},
\frac{1+s+\frac{9}{32\Delta}}{2}},
\]
where $q' = O(k^3q^2\Delta^2)$. Consequently,
\[
\mathsf{StoqMA}(k) = \mathsf{StoqMA}.
\]
\end{theorem}

The new verifier uses $R$ copies of every witness register. 
Honest tensor powers are fixed exactly by the symmetrization, so completeness suffers only the affine change inherent in the branch-overlap formulation.
Approximation is charged entirely to soundness, and the transformed gap remains at least $23/(64\Delta)$.  A polynomially accurate dyadic approximation to the spectral relaxation is the Hermitian overlap matrix of a polynomial-size one-witness branch-overlap verifier, and hence yields a stoquastic verifier.

The preservation of near-perfect completeness lets us combine the collapse with the near-perfect completeness result of Aharonov,
Grilo, and Liu~\cite{AharonovGriloLiu2025}.

\begin{corollary}\label{cor:near-perfect-MA}
Let $k, \Delta : \N \to \N_{\ge 1}$ be polynomially bounded and polynomial-time computable, and let $c,s : \N \to [\frac12,1]\cap\Q$ be polynomial-time computable functions satisfying
$c-s \ge 1/{\Delta}$. If $1-c$ is negligible, then
\[
\mathsf{StoqMA}(k,c,s) \subseteq \mathsf{MA}.
\]
\end{corollary}

General de Finetti theorems ask for approximation in trace distance against every bounded observable, and therefore pay for that
generality by a polynomial factor in the dimension.
Local-measurement theorems ask only for approximation against operationally
restricted measurements.
Multi-prover stoquastic verifiers impose a different restriction, that a test may
be global across all registers but all of its matrix entries are
non-negative.
The desired approximation is correspondingly more structured than an
arbitrary separable state, it must be a mixture of product states whose
local vectors have non-negative coordinates.
Establishing a bound with logarithmic dimension dependence for this test
set, i.e.,~\Cref{thm:positive-de-Finetti}, is the analytic core of the present work.

One must distinguish non-negativity of the witnesses from non-negativity of
the acceptance effect. Stoquastic Merlin-Arthur verification has both
properties. The class $\mathsf{QMA}^{+}$, is defined such that the verifier
remains an arbitrary quantum circuit but the witness is restricted to non-negative states. For suitable constant completeness and soundness parameters, this model
equals $\mathsf{NEXP}$~\cite{BassirianFeffermanMarwaha2024}. 
The two-unentangled-witness variant likewise satisfies $\mathsf{QMA}^{+}(2)=\mathsf{NEXP}$~\cite{JeronimoWu2023}.
Thus restricting the admissible witnesses does not itself remove destructive interference from the verification
circuit, as a general quantum verifier may introduce phases internally, whereas
stoquastic verification additionally forces the compressed acceptance effect
itself to be entrywise non-negative.

It can be observed that the Sum-of-Squares rounding algorithm of~\cite{BarakKelnerSteurer2014} yields a classical exponential time algorithm for $\mathsf{StoqMA}(k)$. 
A common progression in algorithm design is to use a Sum-of-Squares analysis as a guide for a simpler spectral relaxation. 
Spectral algorithms inspired by SoS proofs can retain the essential moment or conditioning argument while replacing a large semidefinite program by an eigenvalue computation~\cite{HopkinsSchrammShiSteurer2016}. 
We follow this progression here. 
The SoS argument of Barak, Kelner, and Steurer is converted into a positive de
Finetti theorem, and the associated relaxation is realized as the largest eigenvalue of an explicit operator on symmetric powers. 
Because this operator remains entrywise nonnegative, it can in turn be approximated by a one-witness stoquastic verifier. 
The resulting verifier transformation proves a collapse, not merely
an exponential-time upper bound. 
In this sense, the product-state constraint itself can be absorbed into
stoquastic verification. This illustrates the general principle that without destructive interference, entanglement becomes tame.  

\begin{figure}[t]
\centering
\resizebox{\textwidth}{!}{%
\begin{tikzpicture}[
  x=1cm,y=1cm,
  cls/.style={draw,rounded corners,align=center,minimum width=1.75cm,minimum height=.62cm,inner sep=3pt},
  emphcls/.style={draw,very thick,rounded corners,align=center,minimum width=3.00cm,minimum height=1.00cm,inner sep=4pt,fill=blue!4},
  arr/.style={-{Latex[length=2mm]},thick},
  oldarr/.style={-{Latex[length=2mm]},thick,dashed},
  every node/.style={font=\small}
]
  \node[cls] (ma) at (0,0) {$\ma$};
  \node[emphcls] (stoq) at (3.25,0)
    {$\begin{gathered}\stoqma\\=\stoqma(2)\\=\stoqma(k)\end{gathered}$};
  \node[cls] (am) at (8,0.75) {$\am$};
  \node[cls] (qma) at (6.75,-0.75) {$\qma$};
  \node[cls] (pp) at (9.25,-0.75) {$\pp$};
  \node[cls] (pspace) at (12,0) {$\pspace$};
  \node[cls] (exp) at (14.5,0) {$\expclass$};
  \node[cls] (nexp) at (17,0) {$\nexp$};
  \node[cls] (qma2) at (8.0,2.25) {$\qma(2)$};

  \draw[arr] (ma) -- (stoq);
  \draw[arr] (stoq.east) to[out=28,in=180] (am.west);
  \draw[arr] (stoq.east) to[out=-28,in=180] (qma.west);
  \draw[arr] (qma.east) -- (pp.west);
  \draw[arr] (am.east) to[out=0,in=150] (pspace.north west);
  \draw[arr] (pp.east) to[out=0,in=210] (pspace.south west);
  \draw[arr] (pspace) -- (exp);
  \draw[arr] (exp) -- (nexp);
  \draw[arr] (stoq.north east) to[out=48,in=190] (qma2.west);
  \draw[arr] (qma2.east) to[out=0,in=150] (nexp.north west);
  \draw[oldarr] (stoq.south east) .. controls +(1.1,-2.25) and +(-2.8,-2) .. (exp.south west);
\end{tikzpicture}%
}
\caption{Complexity landscape around the result.  The highlighted equality is proved here for inverse-polynomial promise gaps by the direct $k$-prover in~\Cref{thm:stoquastic-collapse}.  The dashed arrow is the direct BKS sum-of-squares simulation for nonnegative forms \cite{BarakKelnerSteurer2014}, which gives an $\expclass$ upper bound in the exponentially large witness dimension and is superseded, for stoquastic verifiers, by the collapse.  The containment $\stoqma\subseteq\am$ follows from the stoquastic largest-eigenvalue/approximate-set-size framework of Bravyi--DiVincenzo--Oliveira--Terhal together with the original stoquastic-verification framework \cite{BravyiBessenTerhal2006,BravyiDiVincenzoOliveiraTerhal2008}; and the containment $\qma \subseteq \pp$ is proven in \cite{MarriottWatrous2005}.}
\label{fig:landscape}
\end{figure}
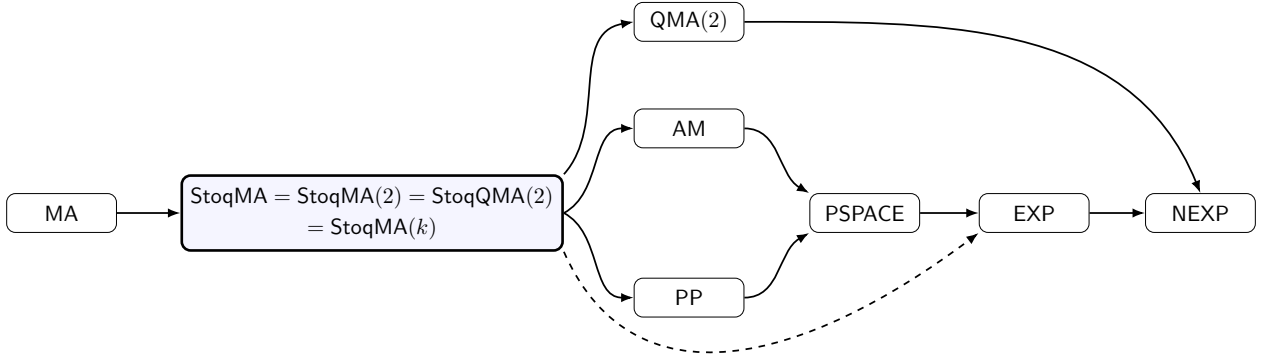

\subsection{Proof Strategy}

The exponential Sum-of-Squares algorithm for $\mathsf{StoqMA}(k)$ of Barak, Kelner, and
Steurer gives a way to relax the product-state condition and then round back
to a product state \cite{BarakKelnerSteurer2014}. We use a similar idea, but
choose a relaxation that can itself be run by a stoquastic verifier. We first
show that symmetric states look like mixtures of nonnegative product states
to any nonnegative test. We then turn this statement into a one-witness
verifier.

To prove the positive de Finetti theorem, consider a state that is symmetric
within each block and keep one copy from each block. The diagonal entries of
the tested state give a classical probability distribution. The PSD condition of
the quantum state bounds each off-diagonal entry using the same probabilities. Since
the acceptance matrix has no negative entries, replacing the tested state by
the pure state whose amplitudes are the square roots of these probabilities
can only increase acceptance. This is where the absence of destructive
interference matters. If the distribution were a product distribution, then
its square-root state would already be a nonnegative product state. We only
need to reduce the correlations in the distribution.

We use the remaining copies to reduce these correlations. At each step, we
pick a block that is most correlated with the others, measure one of its
copies, and condition on the outcome. Another copy from the same block then
becomes the tested copy. Symmetry ensures that the conditional state is still
separately bosonic and that the tested state is unchanged on average. Each
measurement also lowers the entropy on average by at least the correlation
revealed by that measurement. Entropy cannot drop below zero, so the
correlations cannot stay large through many steps. If we stop at a random
step, the remaining registers have
little correlation on average, while averaging over all measurement outcomes
recovers the original tested state. Using \(R\) copies of every block makes
this possible for \(R-1\) steps even if every step measures the same block.

For each possible sequence of outcomes, we compare the square-root state of
the joint distribution with the square-root state of the product of its
marginals. The entropy bound shows that these states are close, so any bounded
nonnegative test gives them nearly the same value. The second state is a
nonnegative product state. Averaging these product states gives the state in
\Cref{thm:positive-de-Finetti}. The measurements do not depend on the test, so
the same state works for every nonnegative test. The error depends on the
initial entropy, which is at most the logarithm of the dimension.

We next turn this approximation into an eigenvalue problem. Take \(R\) copies
of each witness register, run the original acceptance operator on one copy
from each block, and symmetrize before and after the test. Repeating an honest
product witness \(R\) times is unaffected by the symmetrization, so the new
largest eigenvalue is at least the original product-state optimum. In the
other direction, a largest-eigenvalue state is symmetric within each block,
and the positive de Finetti theorem shows that its value is at most the
original optimum plus a small error. More copies make the approximation
tighter, since one untested copy can be discarded from every block without
changing the tested registers. Thus the large Sum-of-Squares relaxation is
replaced by a largest-eigenvalue calculation, following the general approach
of \cite{HopkinsSchrammShiSteurer2016}.

The next issue is implementing the symmetrization. Ideally, we would choose a
uniformly random permutation of the \(R\) copies. The verifier's random choices
come from uniformly random bits, however, and \(1/R!\) usually does not have a
power-of-two denominator. We use the bits to choose an almost-uniform
permutation and then choose equally between that permutation and its inverse.
This keeps the average self-adjoint and entrywise nonnegative. It also leaves
every repeated product witness unchanged, since every permutation fixes such
a witness.

Finally, we use the branch-overlap view of stoquastic verification, in which
the acceptance matrix is the overlap of two reversible computations. We
independently permute the copies in the two computations. Their overlap then
places one approximate symmetrizer on each side of the original acceptance
matrix, giving a polynomial-size one-witness verifier. Its witness contains
\(R\) copies of all \(k\) original witness registers. Honest repeated
witnesses are unchanged. For any other witness, the positive de Finetti
theorem bounds the ideal symmetrization, and the almost-uniform permutations
add only a small further error. Taking \(R\) of order \(k^2q\Delta^2\) and
using enough random bits makes both errors smaller than a fixed fraction of
the original gap. This proves \Cref{thm:stoquastic-collapse}. If the original
completeness is near one, the new verifier still rejects honest witnesses
only with negligible probability and keeps an inverse-polynomial soundness
gap. The result of Aharonov, Grilo, and Liu
\cite{AharonovGriloLiu2025} then gives \Cref{cor:near-perfect-MA}.

%% file: section2-preliminaries.tex
\section{Preliminaries}\label{sec:preliminaries}

All proofs of the claims stated in this section are deferred to \Cref{app:preliminaries}. 

\subsection{Hilbert Spaces and Quantum States}

All Hilbert spaces are finite-dimensional complex spaces equipped with a fixed orthonormal basis. The Hilbert space associated with an $n$-qubit register is $(\C^2)^{\otimes n}$ equipped with the standard computational product basis. An operator $M$ is \emph{non-negative} if all of its matrix entries in the fixed basis are non-negative real numbers. This is distinct from positive semidefiniteness, which is denoted by $M \succeq 0$. We write
\[
\norm{u}_2 := \sqrt{\braket{u}{u}} \qq
\norm{M}_2 := \max_{\norm{u}_2=1}\norm{Mu}_2
\]
for the Hilbert space norm and operator norm. For a Hilbert space $\mcH$, 
we denote the sets of all pure and mixed states
\[
\bbS(\mcH) := \set{\ket{u} \in \mcH : \norm{u}_2=1},
\qq
\bbD(\mcH) := \set{\rho \in \mathrm{End}(\mcH) : \rho \succeq 0 \text{ and } \Tr(\rho) = 1}.
\]
We also define the set of non-negative pure states
\[
\bbS_+(\mcH) := \set{\ket{u}\in\bbS(\mcH):\braket{a}{u}\geq 0\text{ for every basis vector }\ket{a}},
\]
and, for Hilbert spaces $\mcH_1, \cdots, \mcH_k$, the set of non-negatively separable states
\[
\mathrm{Sep}_+(\mcH_1,\ldots,\mcH_k)
:=\operatorname{conv}\set*{
\bigotimes_{j=1}^k\ketbra{u_j}{u_j}:\ket{u_j}\in\bbS_+(\mcH_j)}.
\]

Given a Hilbert space $\mcH$, an integer $N\in\N_{\geq1}$ and a permutation $\sigma\in\mfS_N$, define the tensor-factor permutation operator $U_\sigma$ on $\mcH^{\otimes N}$ by
\[
U_\sigma\ket{u_1\otimes\cdots\otimes u_N}
:=
\ket{u_{\sigma^{-1}(1)}\otimes\cdots\otimes u_{\sigma^{-1}(N)}}.
\]
The symmetric subspace and its orthogonal projector are
\[
\Sym^N(\mcH)
:=
\set{\ket{u}\in\mcH^{\otimes N}:U_\sigma\ket{u}=\ket{u}\text{ for every }\sigma\in\mfS_N}
\]
and
\[
\Pi_{\mathrm{sym}}^{(N)}(\mcH)
:=
\frac1{N!}\sum_{\sigma\in\mfS_N}U_\sigma.
\]

\begin{definition}\label{def:separately-bosonic}
Fix an integer $k\in\N_{\geq1}$, powers $N_1,\ldots,N_k\in\N_{\geq1}$, and Hilbert spaces $\mcH_1,\ldots,\mcH_k$. Set
\[
\Pi:=\bigotimes_{j=1}^k\Pi_{\mathrm{sym}}^{(N_j)}(\mcH_j).
\]
A mixed state $\rho \in \bbD(\mcH_1^{\otimes N_1}\otimes\cdots\otimes\mcH_k^{\otimes N_k})$ is \emph{separately bosonic} if $\Pi \rho\Pi=\rho$. In this case, we write
\[
\rho\in\mathrm{SepBos}(\mcH_1,\ldots,\mcH_k;N_1,\ldots,N_k).
\]
Writing $\mcH_j^{(r)}$ for the $r$-th copy of $\mcH_j$ in $\mcH_j^{\otimes N_j}$, define the tested marginal of such a state to be
\[
\rho_1
:=
\rho_{\mcH_1^{(1)}\otimes\cdots\otimes\mcH_k^{(1)}}.
\]
\end{definition}

\begin{claim}\label{clm:bosonic-reduction}
Let $\rho\in\mathrm{SepBos}(\mcH_1,\ldots,\mcH_k;N_1,\ldots,N_k)$, and fix $j\in[k]$ with $N_j\geq2$.
Tracing out any copy in the $j$-th block gives a state in
\[
\mathrm{SepBos}(\mcH_1,\ldots,\mcH_k;N_1,\ldots,N_j-1,\ldots,N_k).
\]
After relabeling any remaining copy in the $j$-th block as its
first copy, the tested marginal of the reduced state remains
$\rho_1$, and for all $\ket{a}\in\mcH_j$ and $\ket{u}\in\Sym^{N_j}(\mcH_j)$,
\[
\paren*{\bra{a}\otimes I}\ket{u}\in\Sym^{N_j-1}(\mcH_j).
\]
\end{claim}

\subsection{Probability}

All logarithms in this subsection have base $2$. For probability distributions $P,Q$ on a finite set $\Omega$, we denote the entropy of $P$
\[
\mathrm H(P):=-\sum_{a\in\Omega}P(a)\log_2P(a).
\]
For a mixed state $\theta$ on some Hilbert space, let $P_\theta$ be its distribution given by measuring $\theta$ in the computational basis. 
We abbreviate
\[
\mathrm H(\theta):=\mathrm H(P_\theta).
\]
For any two distributions $P, Q$, we denote the Hellinger distance by
\[
\mathrm D_{\mathrm H^2}(P,Q)
:=
1-\sum_{a\in\Omega}\sqrt{P(a)Q(a)},
\]
and the KL divergence by
\[
\mathrm D_{\mathrm{KL}}(P\Vert Q)
:=
\sum_{a:P(a)>0}P(a)\log_2\frac{P(a)}{Q(a)},
\]
where $\mathrm D_{\mathrm{KL}}(P\Vert Q):=+\infty$ if $P$ is not absolutely continuous with respect to $Q$.
For jointly distributed finite random variables $X,Y$, write
\[
\mathrm H(X\mid Y)
:=
\sum_{y:\Pr[Y=y]>0}\Pr[Y=y]\,\mathrm H(P_{X\mid Y=y})
\]
and
\[
\mathrm I(X:Y)
:=
\mathrm H(X)+\mathrm H(Y)-\mathrm H(X,Y)
=
\mathrm D_{\mathrm{KL}}(P_{X,Y}\Vert P_XP_Y).
\]
For jointly distributed finite random variables $X,Y,Z$, write
\[
\mathrm I(X:Y\mid Z)
:=
\sum_{z:\Pr[Z=z]>0}\Pr[Z=z]\,
\mathrm D_{\mathrm{KL}}(P_{X,Y\mid Z=z}\Vert
P_{X\mid Z=z}P_{Y\mid Z=z}).
\]
Finally, for jointly distributed $X_1,\ldots,X_k$ with joint distribution $P$, write
\[
\mathrm{TC}(P)
:=
\mathrm D_{\mathrm{KL}}\paren*{P\,\middle\Vert\,\prod_{j=1}^kP_{X_j}}.
\]

\begin{claim}\label{clm:hellinger-kl}
For probability distributions $P,Q$ on the same finite set,
\[
2\mathrm D_{\mathrm H^2}(P,Q)
=\norm{\sqrt P-\sqrt Q}_2^2
\leq\mathrm D_{\mathrm{KL}}(P\Vert Q).
\]
\end{claim}

\begin{claim}\label{clm:total-correlation}
For jointly distributed finite random variables $X_1,\ldots,X_k$ with joint distribution $P$,
\[
\begin{aligned}
0\leq\mathrm{TC}(P)
&=\sum_{j=1}^k\mathrm H(X_j)-\mathrm H(X_1,\ldots,X_k)\\
&=\sum_{j=1}^{k-1}\mathrm I(X_j:X_{j+1},\ldots,X_k)
\leq\sum_{j=1}^{k-1}
\mathrm I(X_j:X_1,\ldots,X_{j-1},X_{j+1},\ldots,X_k).
\end{aligned}
\]
\end{claim}

\subsection{Stoquastic Verification}

Let $k,q:\N\to\N_{\geq1}$ be polynomially bounded and polynomial-time computable. A $(k,q)$-stoquastic verifier is a family $V=\set{V_x:x\in\bits^*}$ of quantum circuits which is uniformly computable from $x$ in polynomial time. For $n=\abs{x}$, the circuit $V_x$ acts on $k(n)$ witness registers $W_1,\ldots,W_{k(n)}$ of $q(n)$ qubits each together with polynomially many ancillary qubits initialized to $\ket{0}$ or $\ket{+}$. One circuit qubit is designated as the output qubit $O$. The circuit consists only of $\mathsf X$, $\mathsf{CX}$, and $\mathsf{CCX}$ gates and accepts when measuring $O$ in the Hadamard basis gives outcome $+$. If $a_0$ and $a_+$ are the numbers of zero and plus ancillas, set
\[
\ket{\eta}:=\ket{0}^{\otimes a_0}
\otimes\ket{+}^{\otimes a_+},
\qquad
\Pi_+:=\ketbra{+}{+}_O\otimes I,
\]
and define the compressed acceptance matrix on $W:=W_1\otimes\cdots\otimes W_{k(n)}$ by
\begin{equation}\label{eq:compressed-acceptance}
M_x
:=
\paren*{I_W\otimes\bra{\eta}}V_x^\dagger\Pi_+V_x
\paren*{I_W\otimes\ket{\eta}}.
\end{equation}

\begin{definition}\label{def:stoqma-parameters}
Let $c,s:\N\to[\frac12,1]\cap\Q$ be polynomial-time computable and satisfy $s<c$. A promise problem $L=(L_{\mathrm{yes}},L_{\mathrm{no}})$ belongs to $\mathsf{StoqMA}(k,q,c,s)$ if it has a $(k,q)$-stoquastic verifier such that, for every $x\in\bits^*$ of length $n$,
\[
\begin{array}{lll}
x\in L_{\mathrm{yes}}
&\Longrightarrow&
\displaystyle
\max_{\substack{
  \ket{u_j}\in\bbS((\C^2)^{\otimes q(n)})\\
  j\in[k(n)]
}}
\matrixel{u_1\otimes\cdots\otimes u_{k(n)}}{M_x}
{u_1\otimes\cdots\otimes u_{k(n)}}\geq c(n),\\[2ex]
x\in L_{\mathrm{no}}
&\Longrightarrow&
\displaystyle
\max_{\substack{
  \ket{u_j}\in\bbS((\C^2)^{\otimes q(n)})\\
  j\in[k(n)]
}}
\matrixel{u_1\otimes\cdots\otimes u_{k(n)}}{M_x}
{u_1\otimes\cdots\otimes u_{k(n)}}\leq s(n).
\end{array}
\]
Define
\[
\mathsf{StoqMA}(k,c,s)
:=
\bigcup_q\mathsf{StoqMA}(k,q,c,s),
\]
where the union is over polynomially bounded and polynomial-time computable $q$. Define the class $\mathsf{StoqMA}(k)$ by taking the union over polynomial-time computable $c,s$ for which there exists a polynomially bounded and polynomial-time computable $\Delta:\N\to\N_{\geq1}$ satisfying $c-s\geq1/\Delta$. Finally, define
\[
\mathsf{StoqMA}:=\mathsf{StoqMA}(1).
\]
\end{definition}

\begin{proposition}\label{prop:stoquastic-matrix}
For every stoquastic verifier and every $x\in\bits^*$, $M_x$ is a non-negative self-adjoint operator satisfying
\[
0\preceq M_x\preceq I.
\]
Moreover, writing $n:=\abs{x}$,
\[
\begin{aligned}
&\max_{\substack{\ket{u_j}\in\bbS((\C^2)^{\otimes q(n)})\\j\in[k(n)]}}
\matrixel{u_1\otimes\cdots\otimes u_{k(n)}}{M_x}{u_1\otimes\cdots\otimes u_{k(n)}}\\
={}&
\max_{\substack{\ket{u_j}\in\bbS_+((\C^2)^{\otimes q(n)})\\j\in[k(n)]}}
\matrixel{u_1\otimes\cdots\otimes u_{k(n)}}{M_x}{u_1\otimes\cdots\otimes u_{k(n)}}.
\end{aligned}
\]
\end{proposition}

\begin{definition}\label{def:branch-overlap}
A $(k,q)$-branch-overlap verifier is a uniformly generated family of pairs $(C_{0,x},C_{1,x})$ of polynomial-size circuits over $\mathsf X$, $\mathsf{CX}$, and $\mathsf{CCX}$ acting on the same witness and initialized ancillary registers as a $(k,q)$-stoquastic verifier. Its raw compressed overlap and Hermitian overlap are
\[
G_x
:=
\paren*{I_W\otimes\bra{\eta}}C_{0,x}^\dagger C_{1,x}
\paren*{I_W\otimes\ket{\eta}},
\qquad
H_x:=\frac{G_x+G_x^{\mathsf T}}2.
\]
It accepts a witness $\ket{u}\in\bbS(W)$ with probability
\[
\frac12\paren*{1+\matrixel{u}{H_x}{u}}.
\]
Its completeness and soundness values are obtained by maximizing this probability over product witnesses
\[
\ket{u}=\ket{u_1}\otimes\cdots\otimes\ket{u_{k(n)}},
\qquad
\ket{u_j}\in\bbS((\C^2)^{\otimes q(n)}),
\]
where $n:=\abs{x}$.
\end{definition}

\begin{proposition}\label{prop:branch-overlap-equivalence}
For every branch-overlap verifier, $G_x$ and $H_x$ are non-negative,
$H_x$ is self-adjoint, and
\[
\norm{G_x}_2\leq1,
\qquad
\norm{H_x}_2\leq1.
\]
Every $(k,q)$-branch-overlap verifier has a $(k,q)$-stoquastic verifier with the same acceptance probability on every witness, up to polynomially many ancillary qubits and gates. Conversely, if $V_x$ is a $(k,q)$-stoquastic verifier, then the branch-overlap verifier
\[
C_{0,x}:=V_x^\dagger X_O V_x,
\qquad
C_{1,x}:=I
\]
satisfies $(I+H_x)/2=M_x$. Consequently, $\mathsf{StoqMA}(k, q, c, s)$ may be equivalently defined in terms of branch-overlap verifiers. 
\end{proposition}

A function $\nu:\N\to\R_{\geq0}$ is negligible if, for every polynomial
$p:\N\to\R_{>0}$, one has $\nu(n)<1/p(n)$ for all sufficiently large $n$.

\begin{theorem}[{\cite{AharonovGriloLiu2025}}]\label{thm:near-perfect-StoqMA}
Let $c,s:\N\to[\frac12,1]\cap\Q$ be polynomial-time computable. Suppose that $1-c$ is negligible and that there exists a polynomially bounded and polynomial-time computable $\Delta:\N\to\N_{\geq1}$ such that $1-s\geq\frac1\Delta$.
Then,
\[
\mathsf{StoqMA}(1,c,s)\subseteq\mathsf{MA}.
\]
\end{theorem}

%% file: section3-positive-de-finetti.tex
\section{A de Finetti Theorem for Non-Negative Tests}

Throughout this section, fix integers $k,R\in\N_{\ge 2}$ and Hilbert spaces
$\mcH_1,\ldots,\mcH_k$. For $j\in[k]$, let $d_j:=\dim(\mcH_j)$, write $\set{\ket{a} : a \in [d_j]}$ for the computational basis of $\mcH_j$. 
Define $\Omega := [d_1] \times \cdots \times [d_k]$ and write $\set{\ket{\ol{a}} : \ol{a} = (a_1, \cdots, a_k) \in \Omega}$ for the computational product basis of $\mcH_1 \otimes \cdots \otimes \mcH_k$. Set
\[
d:= |\Omega| = \prod_{j=1}^k d_j,
\qquad
L:=\max\set{1,\log_2(d)}.
\]
For every $N \in \N_{\ge 1}$, write
\[
\mcK_N:=\mcH_1^{\otimes N}\otimes\cdots\otimes\mcH_k^{\otimes N},
\qquad
\mcH_j^{\otimes N}=\mcH_j^{(1)}\otimes\cdots\otimes\mcH_j^{(N)},
\]
and define
\[ \Pi_N := \bigotimes_{j=1}^k \Pi_{\rm sym}^{(N)}(\mcH_j). \]
For a self-adjoint operator $M$ on $\mcH_1\otimes\cdots\otimes\mcH_k$, let
\[
\mathrm{SS}_N(M):=\Pi_N\paren*{M\otimes I_{\mathrm{rest}}}\Pi_N,
\]
where $M$ acts on $\mcH_1^{(1)}\otimes\cdots\otimes\mcH_k^{(1)}$, and let
\[
\omega_+(M):=\max_{\substack{\ket{u_j}\in\bbS_+(\mcH_j)\\j\in[k]}}
\matrixel{u_1\otimes\cdots\otimes u_k}{M}{u_1\otimes\cdots\otimes u_k}.
\]
Finally, fix a separately bosonic state
\[
\rho\in\mathrm{SepBos}(\mcH_1,\ldots,\mcH_k;R,\ldots,R).
\]

\subsection{Rounding and Conditioning}

\begin{lemma}\label{lem:direct-rounding}
Let $\theta\in\bbD(\mcH_1\otimes\cdots\otimes\mcH_k),$
let $P:=P_\theta$, and write $p_j$ for its
$j$-th coordinate marginal. Define non-negative unit vectors
$\ket{y_P},\ket{z_P}\in
\mcH_1\otimes\cdots\otimes\mcH_k$ by
\[
\braket{\ol a}{y_P}:=\sqrt{P(\ol a)},
\qquad
\braket{\ol a}{z_P}:=
\sqrt{\prod_{j=1}^k p_j(a_j)}
\]
for every $\ol a=(a_1,\ldots,a_k)\in\Omega$.
Then every non-negative self-adjoint operator $M$ satisfies
\[
\Tr(M\theta)
\leq
\matrixel{z_P}{M}{z_P}
+2\norm{M}_2\sqrt{\mathrm{TC}(P)}.
\]
\end{lemma}
\begin{proof}
Because $\theta\succeq0$, positivity of every $2\times2$
principal minor gives, for all $\ol a,\ol b\in\Omega$,
\[
\abs{\matrixel{\ol a}{\theta}{\ol b}}
\leq
\sqrt{P(\ol a)P(\ol b)}.
\]
Since $M$ is entrywise non-negative and $\Tr(M\theta)$ is real,
\begin{align*}
\Tr(M\theta)
&=
\sum_{\ol a,\ol b}
\matrixel{\ol a}{M}{\ol b}
\operatorname{Re}\matrixel{\ol b}{\theta}{\ol a}
\leq
\sum_{\ol a,\ol b}
\matrixel{\ol a}{M}{\ol b}
\sqrt{P(\ol a)P(\ol b)}
=
\matrixel{y_P}{M}{y_P}.
\end{align*}
Moreover,
\[
\abs{\matrixel{y_P}{M}{y_P}-\matrixel{z_P}{M}{z_P}}
\le 2\norm{M}_2\norm{y_P-z_P}_2.
\]
By \Cref{clm:hellinger-kl},
\[
\norm{y_P-z_P}_2^2
=2\mathrm D_{\mathrm H^2}\paren*{P,\prod_{j=1}^k p_j}
\le \mathrm D_{\mathrm{KL}}\paren*{P\,\middle\Vert\,\prod_{j=1}^k p_j}
=\mathrm{TC}(P).
\]
Combining the preceding inequalities proves the result.
\end{proof}

\begin{lemma}\label{lem:one-conditioning-step}
Fix $r_1,\ldots,r_k\in\N_{\ge 1}$ and
\[
\sigma\in\mathrm{SepBos}(\mcH_1,\ldots,\mcH_k;r_1,\ldots,r_k).
\]
Suppose $r_i\ge 2$ for some $i\in[k]$. Measure $\mcH_i^{(1)}$ in the
computational basis. For every outcome $a$ of positive probability $w_a$, let
\[
\sigma^a:=\frac{1}{w_a}
\paren*{\bra{a}_{\mcH_i^{(1)}}\otimes I}\sigma
\paren*{\ket{a}_{\mcH_i^{(1)}}\otimes I},
\]
and relabel $\mcH_i^{(2)},\ldots,\mcH_i^{(r_i)}$ as
$\mcH_i^{(1)},\ldots,\mcH_i^{(r_i-1)}$. Then every $\sigma^a$ is separately
bosonic and
\begin{align}
\sum_a w_a(\sigma^a)_1&=\sigma_1,\label{eq:conditioning-martingale}\\
\sum_a w_a\mathrm H((\sigma^a)_1)
&\le \mathrm H(\sigma_1)-\mathrm I(X_i:X_{\bar i}),
\label{eq:conditioning-entropy}
\end{align}
where $(X_1,\ldots,X_k)\sim P_{\sigma_1}$ and $X_{\bar i}$ denotes all
coordinates other than $X_i$.
\end{lemma}
\begin{proof}
If $v\in\Sym^{r_i}(\mcH_i)$, then
$\bra{a}_{\mcH_i^{(1)}}v\in\Sym^{r_i-1}(\mcH_i)$. Applying this observation
to a spectral decomposition of $\sigma$ proves that every $\sigma^a$ is
separately bosonic.

Let $X_i'$ be the outcome obtained by measuring the fresh copy
$\mcH_i^{(2)}$. The tested distribution of $\sigma^a$ is the conditional law
of $(X_i',X_{\bar i})$ given $X_i=a$. Separate bosonicity implies invariance
under swapping $\mcH_i^{(1)}$ and $\mcH_i^{(2)}$, so the marginal law of
$(X_i',X_{\bar i})$ is $P_{\sigma_1}$. Therefore
\begin{align*}
\sum_a w_a\mathrm H((\sigma^a)_1)
&=\mathrm H(X_i',X_{\bar i}\mid X_i)\\
&=\mathrm H(\sigma_1)-\mathrm I(X_i:X_i',X_{\bar i})\\
&\le \mathrm H(\sigma_1)-\mathrm I(X_i:X_{\bar i}),
\end{align*}
which proves \Cref{eq:conditioning-entropy}. Summing the unnormalized
post-measurement states traces out $\mcH_i^{(1)}$ and uses
$\mcH_i^{(2)}$ as the tested copy. Swap invariance then gives
\Cref{eq:conditioning-martingale}.
\end{proof}

\begin{lemma}\label{lem:low-correlation-ensemble}
For every integer $T$ with $1\le T\le R-1$, there is a finite ensemble
$\set{(\nu_s,\sigma_s)}_s$ of separately bosonic states, possibly with
different copy numbers, such that
\begin{align}
\sum_s\nu_s(\sigma_s)_1&=\rho_1,
\label{eq:ensemble-marginal}\\
\sum_s\nu_s\mathrm{TC}(P_{(\sigma_s)_1})
&\le \frac{k\mathrm H(\rho_1)}{T}.
\label{eq:ensemble-correlation}
\end{align}
\end{lemma}
\begin{proof}
Set $\rho^{(0)}:=\rho$. At stage $t<T$, for the tested distribution of the
current state define
\[
\mathrm I_*\paren*{\rho^{(t)}}:=\max_{i\in[k]}\mathrm I(X_i:X_{\bar i}).
\]
Choose an index attaining the maximum, perform the operation in
\Cref{lem:one-conditioning-step}, and sample its outcome to obtain
$\rho^{(t+1)}$. Before stage $t$, every extended block has at least
$R-t\ge 2$ copies. By \Cref{eq:conditioning-martingale},
\[
\EE\cbrak*{(\rho^{(t+1)})_1}{\rho^{(t)}}=(\rho^{(t)})_1,
\qquad
\EE(\rho^{(t)})_1=\rho_1.
\]
By \Cref{eq:conditioning-entropy},
\[
\EE\cbrak*{\mathrm H((\rho^{(t+1)})_1)}{\rho^{(t)}}
\le \mathrm H((\rho^{(t)})_1)-\mathrm I_*\paren*{\rho^{(t)}}.
\]
Telescoping and using non-negativity of entropy gives
\begin{equation}\label{eq:entropy-telescope}
\sum_{t=0}^{T-1}\EE\mathrm I_*\paren*{\rho^{(t)}}\le \mathrm H(\rho_1).
\end{equation}
By \Cref{clm:total-correlation}, for every distribution $P$ on $k$ coordinates,
\[
\mathrm{TC}(P)
\le \sum_{i=1}^{k-1}\mathrm I(X_i:X_{\bar i})
\le k\max_{i\in[k]}\mathrm I(X_i:X_{\bar i}).
\]
Let $J$ be uniform on $\set{0,\ldots,T-1}$ and independent of the process,
and let the ensemble be the law of $\rho^{(J)}$. The martingale identity gives
\Cref{eq:ensemble-marginal}, while \Cref{eq:entropy-telescope} gives
\[
\EE\mathrm{TC}(P_{(\rho^{(J)})_1})
\le\frac{k}{T}\sum_{t=0}^{T-1}\EE\mathrm I_*\paren*{\rho^{(t)}}
\le\frac{k\mathrm H(\rho_1)}{T},
\]
which is \Cref{eq:ensemble-correlation}.
\end{proof}

\begin{theorem*}[Restatement of \Cref{thm:positive-de-Finetti}] 
There exists a state $\tau \in \mathrm{Sep}_+(\mcH_1, \cdots, \mcH_k)$ such that
\[
\Tr\paren[\big]{M(\rho_1 - \tau)} \le 2\norm{M}_2 \sqrt{\frac{k\, \rmH(\rho_1)}{R-1}},
\]
simultaneously for all non-negative self-adjoint operators $M$.
\end{theorem*}

\begin{proof}
Apply \Cref{lem:low-correlation-ensemble} with $T=R-1$. For every $s$, let
$\ket{z_s}$ be the non-negative product vector obtained from the square roots
of the one-coordinate marginals of $P_{(\sigma_s)_1}$, as in
\Cref{lem:direct-rounding}, and define
\[
\tau:=\sum_s\nu_s\ketbra{z_s}{z_s}.
\]
Then $\tau\in\mathrm{Sep}_+(\mcH_1,\ldots,\mcH_k)$ and, for every
non-negative self-adjoint $M$,
\begin{align*}
\Tr(M\rho_1)
&=\sum_s\nu_s\Tr(M(\sigma_s)_1)\\
&\le \Tr(M\tau)
+2\norm{M}_2\sum_s\nu_s\sqrt{\mathrm{TC}(P_{(\sigma_s)_1})}\\
&\le \Tr(M\tau)
+2\norm{M}_2\sqrt{\sum_s\nu_s\mathrm{TC}(P_{(\sigma_s)_1})}\\
&\le \Tr(M\tau)
+2\norm{M}_2\sqrt{\frac{k\mathrm H(\rho_1)}{R-1}}.
\end{align*}
The ensemble and $\tau$ depend only on $\rho$, so the inequality holds
simultaneously for all such $M$.
\end{proof}

\subsection{Symmetric Lifts}

\begin{lemma}\label{lem:product-lift}
For every $N\in\N_{\ge 1}$ and every non-negative self-adjoint $M$,
\[
\omega_+(M)\le\lambda_{\max}(\mathrm{SS}_N(M)).
\]
\end{lemma}
\begin{proof}
For non-negative unit vectors $\ket{u_j}\in\mcH_j$, the vector
\[
\ket{\widehat u}:=\ket{u_1}^{\otimes N}\otimes\cdots\otimes
\ket{u_k}^{\otimes N}
\]
is a unit vector fixed by $\Pi_N$, and
\[
\matrixel{\widehat u}{\mathrm{SS}_N(M)}{\widehat u}
=\matrixel{u_1\otimes\cdots\otimes u_k}{M}{u_1\otimes\cdots\otimes u_k}.
\]
Taking the maximum proves the claim.
\end{proof}

\begin{lemma}\label{lem:nested-bosonic-marginals}
If
\[
\sigma\in\mathrm{SepBos}(\mcH_1,\ldots,\mcH_k;N+1,\ldots,N+1),
\]
then tracing out one untested copy from every block gives
a state
\[
\widetilde\sigma\in
\mathrm{SepBos}(\mcH_1,\ldots,\mcH_k;N,\ldots,N)
\]
such that $\widetilde\sigma_1=\sigma_1$.
\end{lemma}
\begin{proof}
For $v\in\Sym^{N+1}(\mcH_j)$ and every computational-basis vector
$\ket a\in\mcH_j$, contraction of any tensor factor gives
$\bra a v\in\Sym^N(\mcH_j)$. Expanding the partial trace in the
computational basis and applying this observation to a spectral decomposition
of $\sigma$ proves the support statement. Only untested copies are traced out,
so the tested marginal is unchanged.
\end{proof}

\begin{theorem*}[Restatement of \Cref{thm:symmetric-lift}]
For every non-negative self-adjoint operator $M$ on $\mcH_1 \otimes \cdots \otimes \mcH_k$,
\[
0 \le
\lambda_{\max}\paren*{\mathrm{SS}_R(M)}
- \max_{\substack{\ket{u_j} \in \bbS_+(\mcH_j) \\ j \in [k]}}
\matrixel{u_1 \otimes \cdots \otimes u_k}{M}{u_1 \otimes \cdots \otimes u_k}
\le
2\norm{M}_2\sqrt{\frac{k\max\set{1, \log_2(d)}}{R-1}}.
\]
In particular, if $\norm{M}_2 \le 1$, $\eps \in (0,1)$, and
\[
R \ge 1 + \frac{4k\max\set{1, \log_2(d)}}{\eps^2},
\]
then the upper bound is at most $\eps$.
\end{theorem*}

\begin{proof}
By \Cref{lem:product-lift},
\[
\omega_+(M)\le\lambda_{\max}(\mathrm{SS}_R(M)).
\]
For every $N\ge 1$, the variational principle and the preceding inequality give
\begin{equation}\label{eq:bosonic-variational}
\lambda_{\max}(\mathrm{SS}_N(M))
=\max_{\substack{\sigma\in\bbD(\mcK_N)\\
\Pi_N\sigma\Pi_N=\sigma}}
\Tr(M\sigma_1).
\end{equation}
Indeed, $\mathrm{SS}_N(M)$ vanishes on the orthogonal complement of the range
of $\Pi_N$, and its largest eigenvalue on that range is at least
$\omega_+(M)\ge0$, where the last inequality follows by evaluating $M$ on a
computational-basis product vector. For every feasible $\sigma$,
\[
\Tr(\mathrm{SS}_N(M)\sigma)
=\Tr\paren*{(M\otimes I_{\mathrm{rest}})\sigma}
=\Tr(M\sigma_1).
\]

Let $\sigma$ attain the maximum in \Cref{eq:bosonic-variational} for $N=R$.
Applying \Cref{thm:positive-de-Finetti} to $\sigma$ gives a state
$\tau\in\mathrm{Sep}_+(\mcH_1,\ldots,\mcH_k)$ such that
\begin{align*}
\lambda_{\max}(\mathrm{SS}_R(M))
&=\Tr(M\sigma_1)\\
&\le\Tr(M\tau)
+2\norm{M}_2\sqrt{\frac{k\mathrm H(\sigma_1)}{R-1}}\\
&\le\omega_+(M)
+2\norm{M}_2\sqrt{\frac{kL}{R-1}},
\end{align*}
where $\Tr(M\tau)\le\omega_+(M)$ and
$\mathrm H(\sigma_1)\le\log_2(d)\le L$.

By \Cref{lem:nested-bosonic-marginals}, every state feasible in
\Cref{eq:bosonic-variational} at level $R+1$ has a level-$R$ marginal with
the same tested state. Hence
\[
\lambda_{\max}(\mathrm{SS}_{R+1}(M))
\le\lambda_{\max}(\mathrm{SS}_R(M)).
\]
Finally, if $\norm{M}_2\le1$ and
$R-1\ge4kL/\eps^2$, then
\[
2\norm{M}_2\sqrt{\frac{kL}{R-1}}\le\eps.
\qedhere
\]
\end{proof}

%% file: section4-stoquastic-collapse.tex
\section{A Collapse for Multi-Witness Stoquastic Verification}

Throughout this section, for integers $k \ge 2$, $q \ge 1$, and $R \ge 2$, write
\[
\mcH := (\C^2)^{\otimes q}
\qquad\text{and}\qquad
\mcK_R := \paren*{\mcH^{\otimes R}}^{\otimes k}.
\]
For an operator $M$ on $\mcH^{\otimes k}$, let $M^{(1)}$ denote its action on the first copy of each of the $k$ factors of $\mcK_R$, tensored with the identity on all remaining copies.

\subsection{Dyadic Symmetrization}

\begin{lemma}\label{lem:dyadic-symmetrizer}
Fix $R \ge 2$ and a rational number $\eta \in (0,1)$.
There exist an integer
\[
b = 1 + \ceil*{\log_2\paren*{R!/\eta}}
\]
and a polynomial-time computable map $g : \bits^b \to \mfS_R$ such that the distribution
\[
\mu(\sigma) := 2^{-b}\abs*{\set{z \in \bits^b : g(z)=\sigma}}
\]
satisfies
\[
\mu(\sigma)=\mu(\sigma^{-1})
\qquad\text{and}\qquad
\sum_{\sigma\in\mfS_R}\abs*{\mu(\sigma)-\frac1{R!}} \le \eta.
\]
Consequently, for every Hilbert space $\mcG$, the operator
\[
\wt\Pi_{R,\eta}(\mcG)
:= 2^{-b}\sum_{z\in\bits^b}U_{g(z)}
\]
is non-negative and self-adjoint, satisfies
\[
\norm*{\wt\Pi_{R,\eta}(\mcG)}_2 \le 1,
\qquad
\norm*{\wt\Pi_{R,\eta}(\mcG)-\Pi_{\rm sym}^{(R)}(\mcG)}_2 \le \eta,
\]
and fixes $\ket{u}^{\otimes R}$ for every $\ket{u}\in\mcG$.
If $\mcG=(\C^2)^{\otimes q}$, the controlled operation
\[
\ket{z}\ket{v}\longmapsto \ket{z}U_{g(z)}\ket{v}
\]
on $(\C^2)^{\otimes(b+Rq)}$ is implementable uniformly by polynomially many $\mathsf X$, $\mathsf{CX}$, and $\mathsf{CCX}$ gates and polynomially many zero ancillas.
\end{lemma}

\begin{proof}
Set $N:=R!$, $b_0:=\ceil*{\log_2(N/\eta)}$, and $B:=2^{b_0}$.
Fix the factoradic bijection
\[
\pi:\set{0,1,\ldots,N-1}\longrightarrow\mfS_R
\]
and define $h(z):=\pi(z\bmod N)$ for $z\in\bits^{b_0}$, regarded as an integer in $\set{0,1,\ldots,B-1}$.
If
\[
\nu(\sigma):=2^{-b_0}\abs*{\set{z\in\bits^{b_0}:h(z)=\sigma}},
\]
then every fiber of $h$ has size either $\floor{B/N}$ or $\ceil{B/N}$.
It follows that
\[
\sum_{\sigma\in\mfS_R}\abs*{\nu(\sigma)-\frac1N}
<\frac{N}{B}\le\eta.
\]
For $(z,t)\in\bits^{b_0}\times\bits$, define
\[
g(z,t):=
\begin{cases}
h(z),&t=0,\\
h(z)^{-1},&t=1.
\end{cases}
\]
Thus $b=b_0+1$ and
\[
\mu(\sigma)=\frac12\paren*{\nu(\sigma)+\nu(\sigma^{-1})}.
\]
Hence $\mu(\sigma)=\mu(\sigma^{-1})$ and
\[
\sum_{\sigma\in\mfS_R}\abs*{\mu(\sigma)-\frac1N}
\le \frac12\sum_{\sigma\in\mfS_R}
\paren*{\abs*{\nu(\sigma)-\frac1N}
+\abs*{\nu(\sigma^{-1})-\frac1N}}
\le\eta.
\]
Since $U_\sigma^\dagger=U_{\sigma^{-1}}$, inverse invariance gives
$\wt\Pi_{R,\eta}(\mcG)^\dagger=\wt\Pi_{R,\eta}(\mcG)$.
Moreover,
\[
\begin{aligned}
\norm*{\wt\Pi_{R,\eta}(\mcG)-\Pi_{\rm sym}^{(R)}(\mcG)}_2
&\le \sum_{\sigma\in\mfS_R}
\abs*{\mu(\sigma)-\frac1N}\norm{U_\sigma}_2
\le\eta,\\
\norm*{\wt\Pi_{R,\eta}(\mcG)}_2
&\le\sum_{\sigma\in\mfS_R}\mu(\sigma)\norm{U_\sigma}_2=1.
\end{aligned}
\]
Every $U_\sigma$ is non-negative and fixes $\ket{u}^{\otimes R}$, which proves the remaining operator properties.
Factoradic unranking, permutation inversion, and the permutation of $R$ blocks of $q$ bits are computable in time polynomial in $R$, $q$, and $b$.
Their reversible implementations, followed by uncomputation of their work registers, give the stated circuit.
\end{proof}

\subsection{Verifier Realization}

\begin{lemma}\label{lem:acceptance-matrix-overlap}
Let $V$ be a stoquastic verifier with witness space $W$, reversible circuit $U$, output qubit $o$, ancilla state $\ket{\alpha}$, and compressed acceptance matrix
\[
M:=\paren*{I_W\otimes\bra{\alpha}}U^\dagger
\paren*{\ketbra{+}{+}_o\otimes I}U
\paren*{I_W\otimes\ket{\alpha}}.
\]
There is a branch-overlap verifier whose raw and Hermitian overlap matrices are both $M$ and which uses one additional plus ancilla.
\end{lemma}

\begin{proof}
Let $r$ be a plus ancilla and define the two branch circuits
\[
A_0:=I_r\otimes U
\qquad\text{and}\qquad
A_1:=\mathsf{CX}_{r,o}(I_r\otimes U).
\]
Since
\[
\bra{+}_r\mathsf{CX}_{r,o}\ket{+}_r
=\frac12(I+X_o)=\ketbra{+}{+}_o,
\]
their compressed overlap matrix is
\[
\paren*{\bra{+}_r\otimes I_W\otimes\bra{\alpha}}A_0^\dagger A_1
\paren*{\ket{+}_r\otimes I_W\otimes\ket{\alpha}}
=\paren*{I_W\otimes\bra{\alpha}}U^\dagger
\paren*{\ketbra{+}{+}_o\otimes I}U
\paren*{I_W\otimes\ket{\alpha}}=M.
\]
Since $M$ is real and self-adjoint, its Hermitian overlap is also $M$.
\end{proof}

For $\eta\in(0,1)\cap\Q$, define
\[
\wt\Pi_{R,\eta}
:=\wt\Pi_{R,\eta}(\mcH)^{\otimes k}
\qquad\text{and}\qquad
\wt{\mathrm{SS}}_{R,\eta}(M)
:=\wt\Pi_{R,\eta}M^{(1)}\wt\Pi_{R,\eta}.
\]

\begin{proposition}\label{prop:approximate-symmetric-lift-verifier}
Let $M$ be the compressed acceptance matrix of a $k$-witness stoquastic verifier with $q$-qubit witness registers.
There is a one-witness branch-overlap verifier on
\[
q'=kRq
\]
witness qubits whose raw and Hermitian overlap matrices are both $\wt{\mathrm{SS}}_{R,\eta}(M)$.
If $b$ is as in \Cref{lem:dyadic-symmetrizer}, the verifier uses $1+2kb$ additional plus ancillas, apart from those of the original verifier, and
\[
\norm*{\wt{\mathrm{SS}}_{R,\eta}(M)-\mathrm{SS}_R(M)}_2
\le 2k\eta.
\]
\end{proposition}

\begin{proof}
Let $A_0,A_1$ be the branch circuits from \Cref{lem:acceptance-matrix-overlap}.
For
\[
\bma=(a_1,\ldots,a_k)
\in(\bits^b)^k,
\]
write
\[
P_{\bma}:=
U_{g(a_1)}\otimes\cdots\otimes U_{g(a_k)}
\]
on $\mcK_R$.
Introduce two independent collections of $kb$ plus ancillas, with basis values $\bma$ and $\bmc$.
The first branch applies $P_{\bma}$ and then $A_0$ to the first copy of each witness register.
The second branch applies $P_{\bmc}$ and then $A_1$ to the same registers.
By \Cref{lem:dyadic-symmetrizer}, these are uniform polynomial-size reversible circuits.
Their compressed overlap matrix is
\[
\begin{aligned}
2^{-2kb}\sum_{\bma,\bmc}
P_{\bma}^\dagger M^{(1)}P_{\bmc}
&=\paren*{2^{-kb}\sum_{\bma}P_{\bma}}^\dagger
M^{(1)}
\paren*{2^{-kb}\sum_{\bmc}P_{\bmc}}\\
&=\wt\Pi_{R,\eta}M^{(1)}\wt\Pi_{R,\eta}
=\wt{\mathrm{SS}}_{R,\eta}(M).
\end{aligned}
\]
This matrix is real and self-adjoint, so it also equals the Hermitian overlap matrix.
Let
\[
\Pi_R:=\Pi_{\rm sym}^{(R)}(\mcH)^{\otimes k}.
\]
A telescoping expansion and \Cref{lem:dyadic-symmetrizer} give
\[
\norm*{\wt\Pi_{R,\eta}-\Pi_R}_2\le k\eta.
\]
Since $\norm{M}_2\le1$ and $\norm{\wt\Pi_{R,\eta}}_2,\norm{\Pi_R}_2\le1$,
\[
\begin{aligned}
\norm*{\wt{\mathrm{SS}}_{R,\eta}(M)-\mathrm{SS}_R(M)}_2
&\le
\norm*{(\wt\Pi_{R,\eta}-\Pi_R)M^{(1)}\wt\Pi_{R,\eta}}_2
+\norm*{\Pi_RM^{(1)}(\wt\Pi_{R,\eta}-\Pi_R)}_2\\
&\le2k\eta.
\end{aligned}
\]
\end{proof}

\subsection{Completeness and Soundness}

\begin{theorem*}[Restatement of \Cref{thm:stoquastic-collapse}]
Let $k, q, \Delta : \N \to \N_{\ge 1}$ be polynomially bounded and polynomial-time computable, with $k \ge 2$, and let $c,s : \N \to [\frac12,1]\cap\Q$ be polynomial-time computable functions satisfying
$c-s \ge 1/{\Delta}$.
Define
\[
R := 1 + \ceil{64k^2q\Delta^2}
\qquad\text{and}\qquad
q' := kRq.
\]
Then
\[
\mathsf{StoqMA}(k,q,c,s)
\subseteq
\mathsf{StoqMA}\paren*{1,q',\frac{1+c}{2},
\frac{1+s+\frac{9}{32\Delta}}{2}},
\]
where $q' = O(k^3q^2\Delta^2)$. Consequently,
\[
\mathsf{StoqMA}(k) = \mathsf{StoqMA}.
\]
\end{theorem*}

\begin{proof}
Fix an input length $n$ and suppress evaluation of $k,q,\Delta,c,s$ at $n$.
Set
\[
\eps:=\frac1{4\Delta},
\qquad
\eta:=\frac1{64k\Delta},
\qquad
R:=1+\ceil*{64k^2q\Delta^2}.
\]
Let $V_x$ be the verifier on an input $x\in\bits^n$, and let $M_x$ be its compressed acceptance matrix.
By \Cref{prop:stoquastic-matrix}, $M_x$ is non-negative and self-adjoint and $\norm{M_x}_2\le1$.
By \Cref{prop:approximate-symmetric-lift-verifier}, $\wt{\mathrm{SS}}_{R,\eta}(M_x)$ is the overlap matrix of a one-witness branch-overlap verifier with a $q'$-qubit witness, where
\[
q'=kRq.
\]
The branch length in \Cref{lem:dyadic-symmetrizer} satisfies
\[
b=O\paren*{R\log R+\log(k\Delta)},
\]
so this family is uniform and has polynomial size.

Suppose $x$ is a yes-instance.
There are witnesses $\ket{u_1},\ldots,\ket{u_k}\in\bbS_+(\mcH)$ such that
\[
\matrixel{u_1\otimes\cdots\otimes u_k}{M_x}
{u_1\otimes\cdots\otimes u_k}\ge c.
\]
Define
\[
\ket{U}:=\ket{u_1}^{\otimes R}\otimes\cdots\otimes
\ket{u_k}^{\otimes R}.
\]
Every permutation in the support of the dyadic symmetrizer fixes $\ket{U}$, and hence
\[
\matrixel{U}{\wt{\mathrm{SS}}_{R,\eta}(M_x)}{U}
=\matrixel{u_1\otimes\cdots\otimes u_k}{M_x}
{u_1\otimes\cdots\otimes u_k}\ge c.
\]
By \Cref{prop:branch-overlap-equivalence}, the resulting stoquastic verifier has completeness at least $(1+c)/2$.

Suppose $x$ is a no-instance.
Since $\log_2\dim(\mcH^{\otimes k})=kq$, \Cref{thm:symmetric-lift} and the choice of $R$ give
\[
\lambda_{\max}\paren*{\mathrm{SS}_R(M_x)}
\le s+2\sqrt{\frac{k^2q}{R-1}}
\le s+\frac1{4\Delta}=s+\eps.
\]
Furthermore,
\[
2k\eta=\frac1{32\Delta}.
\]
By the variational principle and \Cref{prop:approximate-symmetric-lift-verifier},
\[
\lambda_{\max}\paren*{\wt{\mathrm{SS}}_{R,\eta}(M_x)}
\le s+\frac1{4\Delta}+\frac1{32\Delta}
=s+\frac9{32\Delta}.
\]
Applying \Cref{prop:branch-overlap-equivalence}, the soundness of the resulting stoquastic verifier is at most
\[
\frac{1+s+\frac9{32\Delta}}2.
\]
Its completeness-soundness gap is at least
\[
\frac{c-s-\frac9{32\Delta}}2
\ge\frac{23}{64\Delta}.
\]
Finally,
\[
R\le64k^2q\Delta^2+2
\]
and consequently
\[
q'=kRq=O(k^3q^2\Delta^2).
\]
This proves the stated inclusion.
Taking the union over admissible parameters gives
$\mathsf{StoqMA}(k)\subseteq\mathsf{StoqMA}$.
The reverse inclusion follows by ignoring $k-1$ witness registers.
\end{proof}

\begin{corollary*}[Restatement of \Cref{cor:near-perfect-MA}]
Let $k, \Delta : \N \to \N_{\ge 1}$ be polynomially bounded and polynomial-time computable, and let $c,s : \N \to [\frac12,1]\cap\Q$ be polynomial-time computable functions satisfying
$c-s \ge 1/{\Delta}$. If $1-c$ is negligible, then
\[
\mathsf{StoqMA}(k,c,s) \subseteq \mathsf{MA}.
\]
\end{corollary*}

\begin{proof}
Let $L\in\mathsf{StoqMA}(k,c,s)$, and let $q$ be the polynomial witness bound of its verifier.
Define $\ol k:=\max\set{k,2}$. By adding $\ol k-k$ ignored $q$-qubit witness registers, the verifier is an $\ol k$-witness verifier with the same completeness and soundness. By \Cref{thm:stoquastic-collapse}, $L$ has a one-witness stoquastic verifier with
\[
c':=\frac{1+c}{2}
\qquad\text{and}\qquad
s':=\frac{1+s+\frac9{32\Delta}}2.
\]
It holds that
\[
1-c'=\frac{1-c}{2},
\]
which is negligible.
Since $s\le c-1/\Delta\le1-1/\Delta$,
\[
s'\le1-\frac{23}{64\Delta}.
\]
Since $23/64\ge1/3$, the integer-valued function $3\Delta$ is polynomially bounded and polynomial-time computable and
\[
1-s'\ge\frac{23}{64\Delta}\ge\frac1{3\Delta}.
\]
Thus \Cref{thm:near-perfect-StoqMA} implies $L\in\mathsf{MA}$.
\end{proof}

%% file: sectionA-deferred-proofs.tex
\section{Proofs of Preliminary Statements}
\crefalias{section}{appendix}
\label{app:preliminaries}

\begin{proof}[Proof of \Cref{clm:bosonic-reduction}]
It suffices first to prove the assertion about vectors. Let
$\ket u\in\Sym^{N_j}(\mcH_j)$, let $\tau\in\mfS_{N_j-1}$, and let
$\widehat\tau\in\mfS_{N_j}$ fix the first factor and act as $\tau$ on the
remaining factors. Then
\[
U_\tau\paren*{\bra a\otimes I}\ket u
=\paren*{\bra a\otimes I}U_{\widehat\tau}\ket u
=\paren*{\bra a\otimes I}\ket u.
\]
Contraction of any other tensor factor has the same property after a
permutation of the factors.

The condition $\Pi\rho\Pi=\rho$ implies that the range of $\rho$ is contained
in the tensor product of the symmetric subspaces. Take a spectral
decomposition of $\rho$. For each rank-one term, expand the partial trace over
the selected copy in the fixed basis. Every resulting vector is separately
symmetric by the preceding calculation. Hence the reduced state has the
claimed separately bosonic support. 
If the traced copy is untested, the partial trace does not
change the reduced state on the original tested registers.
Permutation invariance within the $j$-th block then allows any
remaining copy to be relabeled as the first copy. The same
invariance handles the case in which the tested copy is traced
out: after relabeling any remaining copy as the first copy, the
tested marginal is again $\rho_1$.
\end{proof}

\begin{proof}[Proof of \Cref{clm:hellinger-kl}]
Set $F(P,Q):=\sum_a\sqrt{P(a)Q(a)}$. If $P\ll Q$, Jensen's inequality and
$-\ln t\geq1-t$ give
\begin{align*}
(\ln2)\mathrm D_{\mathrm{KL}}(P\Vert Q)
&=-2\sum_aP(a)\ln\sqrt{\frac{Q(a)}{P(a)}}\\
&\geq-2\ln F(P,Q)\\
&\geq2\paren*{1-F(P,Q)}.
\end{align*}
Since $\ln2<1$,
\[
\mathrm D_{\mathrm{KL}}(P\Vert Q)
\geq2\paren*{1-F(P,Q)}
=\sum_a\paren*{\sqrt{P(a)}-\sqrt{Q(a)}}^2.
\]
If $P\not\ll Q$, the inequality is immediate.
\end{proof}

\begin{proof}[Proof of \Cref{clm:total-correlation}]
Gibbs' inequality gives $\mathrm{TC}(P)\geq0$. Expanding the logarithm gives
\[
\mathrm{TC}(P)
=\sum_{j=1}^k\mathrm H(X_j)-\mathrm H(X_1,\ldots,X_k).
\]
Moreover,
\begin{align*}
\sum_{j=1}^{k-1}\mathrm I(X_j:X_{j+1},\ldots,X_k)
&=\sum_{j=1}^{k-1}\paren*{
\mathrm H(X_j)+\mathrm H(X_{j+1},\ldots,X_k)
-\mathrm H(X_j,\ldots,X_k)}\\
&=\sum_{j=1}^k\mathrm H(X_j)-\mathrm H(X_1,\ldots,X_k).
\end{align*}
Finally, for every $j\in[k-1]$,
\begin{align*}
&\mathrm I(X_j:X_1,\ldots,X_{j-1},X_{j+1},\ldots,X_k)\\
&\qquad=\mathrm I(X_j:X_{j+1},\ldots,X_k)
+\mathrm I(X_j:X_1,\ldots,X_{j-1}\mid X_{j+1},\ldots,X_k)\\
&\qquad\geq\mathrm I(X_j:X_{j+1},\ldots,X_k).
\end{align*}
Summing over $j$ proves the claim.
\end{proof}

\begin{proof}[Proof of \Cref{prop:stoquastic-matrix}]
Let $J$ be the isometry $J\ket u=\ket u\ket\eta$. Then
\[
M_x=J^\dagger V_x^\dagger\Pi_+V_xJ.
\]
Every $\mathsf X$, $\mathsf{CX}$, and $\mathsf{CCX}$ gate is a permutation
matrix. The vector $\ket\eta$ and the matrix
$\Pi_+=(I+X_O)/2\otimes I$ are non-negative. It follows directly from the
matrix entries in \eqref{eq:compressed-acceptance} that $M_x$ is
non-negative. Since $\Pi_+$ is an orthogonal projector,
\[
0\preceq J^\dagger V_x^\dagger\Pi_+V_xJ\preceq J^\dagger J=I.
\]

Fix a product unit vector $\ket u=\bigotimes_{j=1}^k\ket{u_j}$, and let
$\ket{v_j}$ be obtained by taking the absolute value of every coordinate of
$\ket{u_j}$. Then $\ket{v_j}\in\bbS_+$ and, writing
$\ket v=\bigotimes_j\ket{v_j}$,
\[
\matrixel{u}{M_x}{u}
\leq\abs{\matrixel{u}{M_x}{u}}
\leq\sum_{a,b}(M_x)_{a,b}\abs{u_a}\abs{u_b}
=\matrixel{v}{M_x}{v}.
\]
Taking maxima proves the equality.
\end{proof}

\begin{proof}[Proof of \Cref{prop:branch-overlap-equivalence}]
Let $J\ket u:=\ket u\ket\eta$. Since $C_{0,x}^\dagger C_{1,x}$ is a
permutation matrix,
\[
G_x=J^\dagger C_{0,x}^\dagger C_{1,x}J
\]
is non-negative and has operator norm at most $1$. Thus $H_x$ is
non-negative and self-adjoint, and
\[
\norm{H_x}_2
\leq\frac12\paren*{\norm{G_x}_2+\norm{G_x^{\mathsf T}}_2}
=\norm{G_x}_2\leq1.
\]

Let $B$ be a new qubit initialized to $\ket{+}$. On $B$ and the original
registers, apply
\[
U_x:=\ketbra{0}{0}_B\otimes C_{0,x}+\ketbra{1}{1}_B\otimes C_{1,x}
\]
and designate $B$ as the output qubit. Controlled versions of the gates in
$C_{0,x}$ and $C_{1,x}$ can be implemented by $\mathsf X$, $\mathsf{CX}$,
and $\mathsf{CCX}$ gates with polynomially many zero ancillas. For
$\ket\phi:=\ket u\ket\eta$, the probability of outcome $+$ is
\begin{align*}
\frac14\norm{C_{0,x}\ket\phi+C_{1,x}\ket\phi}_2^2
&=\frac12\paren*{1+\operatorname{Re}
\matrixel{\phi}{C_{0,x}^\dagger C_{1,x}}{\phi}}\\
&=\frac12\paren*{1+\matrixel{u}{H_x}{u}}.
\end{align*}
Here $G_x$ is real, so $(G_x+G_x^{\mathsf T})/2$ is its Hermitian part.

Conversely, set $C_{0,x}:=V_x^\dagger X_OV_x$ and $C_{1,x}:=I$. The first
circuit is self-adjoint, and hence
\[
G_x=H_x
=\paren*{I_W\otimes\bra\eta}V_x^\dagger X_OV_x
\paren*{I_W\otimes\ket\eta}.
\]
Since $\ketbra{+}{+}_O=(I+X_O)/2$,
\[
\frac{I+H_x}{2}
=\paren*{I_W\otimes\bra\eta}V_x^\dagger\Pi_+V_x
\paren*{I_W\otimes\ket\eta}
=M_x.
\]
Both transformations are uniform and have polynomial overhead.
\end{proof}